\documentclass{article}
\usepackage{amssymb,amsfonts}

\def\qed{$\square$}
\let\eps\varepsilon
\let\phi\varphi
\newtheorem{theorem}{Theorem}
\newtheorem{claim}[theorem]{Claim}
\newtheorem{lemma}[theorem]{Lemma}
\newtheorem{DDDefinition}[theorem]{Definition}
\def\enddefinition{\end{DDDefinition}\egroup\medbreak}
\def\proof{\begin{demo}{Proof}}
\def\endproof{\enspace\hfill\qed\end{demo}\medbreak}
\makeatletter
\def\definition{\bgroup\def\@begintheorem##1##2{\trivlist
  \item[\hskip\labelsep{\bfseries ##1\ ##2}]}\begin{DDDefinition}}
\makeatother
\newenvironment{demo}[1]
{\par\medbreak\noindent{\bf #1\enskip}\rm\ignorespaces}{}

\begin{document}
\title{\bf Complexity of universal access structures}
\author{L\'aszl\'o Csirmaz\thanks{This research was partially
supported by the ``Lend\"ulet Program'' of the Hungarian Academy of Sciences
and by the grant NKTH OM-00289/2008}}
\date{\small Central European University, Budapest\\ R\'enyi Institute,
Budapest}
\maketitle

\begin{abstract}
An important parameter in a secret sharing scheme is the number of minimal
qualified sets. Given this number, the universal access structure is the
richest possible structure, namely the one in which there are one or more
participants in every possible Boolean combination of the minimal qualified
sets. Every access structure is a substructure of the universal structure
for the same number of minimal qualified subsets, thus universal access
structures have the highest complexity given the number of minimal qualified
sets. We show that the complexity of the universal structure with $n$
minimal qualified sets is between $n/\log_2n$ and $n/2.7182\dots$
asymptotically.

\noindent{\bf Keywords:} secret sharing; complexity; entropy method;
harmonic series.

\noindent{\bf MSC numbers:} 94A62, 90C25, 05B35.
\end{abstract}

\section{Introduction}

In a secret sharing scheme the access structure defines which subsets of the
participants should recover the secret -- these are the {\it qualified
subsets}. Unqualified subsets are also called {\it independent}. The
collection of qualified subsets is determined uniquely by the {\it minimal}
qualified sets. Suppose in an access structure we have $n$ minimal qualified
sets: $A_1$, $\dots$, $A_n$. For $\eps_i=0$ or $1$ let $A_i^{\eps_i}$ be
$A_i$ when $\eps_i=1$, and the complement of $A_i$ when $\eps_i=0$. The
access structure is {\it universal} if none of the $2^n$ possible
intersections $\bigcap A_i^{\eps_i}$ is empty. In this note $\mathcal U_n$
denotes any universal access structure with $n$ qualified subsets.

Universal structures with at most three minimal qualified sets were
investigated in \cite{ minimal34}, and for four qualified sets in
\cite{minimal4}. For $n=2$ and $n=3$ the exact complexity is known: it is
$1$ and $3/2$, respectively. For $n=4$ it is proved in \cite{minimal4} that
the complexity of $\mathcal U_4$ is between $7/4$ and $11/6$, the exact
value is not known. This paper provides the first bounds for the complexity
of $\mathcal U_n$ for larger values of $n$. The upper bound follows from a
secret sharing scheme defined recursively for each $n$. The lower bound is a
well-tailored application of the {\em independent sequence method}
described, e.g., in \cite{tight}.

\section{Reduction}

In this section we show that all universal structures with $n$ minimal
qualified subsets have the same complexity, and pinpoint a particular
structure which will be investigated later in details. Result of this
section were also proved in \cite{minimal4}.

First we note that the number of participants in each class $\bigcap
A_i^{\eps_i}$ is irrelevant: the complexity of the scheme does not depend on
the number of participants in the classes, it depends only on which classes
are empty and which are not. Consequently, for each $n$, we can define a
unique ``general'' universal access structure with $n$ minimal qualified
sets. As a first step, we can evidently discard all participants who are not
in any of the qualified sets as they can play no role in recovering the
secret. Less trivially we can also discard those participants which are in
all of the qualified sets. Details follow.

Let $\mathcal U^*_n$ be the {\it normalized} structure with $n$ minimal
qualified sets $A_1$, $\dots$, $A_n$, where the intersection $\bigcap
A_i^{\eps_i}$ has exactly one element when not all the $\eps_i$ are equal,
and is empty otherwise. In particular, the structure $\mathcal U^*_n$ is
{\it connected}, no participant is in all, or none, of the minimal qualified
subsets, and has exactly $2^n-2$ participants.

\begin{claim}\label{claim:general}
The complexity of the structures $\mathcal U_n$ and $\mathcal U^*_n$ are
equal: $\sigma(\mathcal U_n) = \sigma(\mathcal U^*_n)$
\end{claim}
\begin{proof}
Let $\mathcal S$ be any scheme realizing $\mathcal U_n$. Pick one
participant from each class $\bigcap A_i^{\eps_i}$ except when all the
$\eps_i$'s are equal, $2^n-2$ participants in total. Make the shares of all
other participants public. This defines a scheme $\mathcal S^*$ realizing
the access structure $\mathcal U^*_n$, and it is clear that the complexity
of $\mathcal S^*$ is less than or equal to the complexity of $\mathcal S$.
This establishes $\sigma(\mathcal U^*_n) \le \sigma(\mathcal U_n)$.

To see the other direction, let $\mathcal S^*$ be a scheme realizing
$\mathcal U^*_n$. For the sake of simplicity we assume that the secret for
$\mathcal S^*$ is a single random bit. Define the scheme $\mathcal S$ for
$\mathcal U_n$ as follows. Write the random bit $s$ as the mod 2 sum of two
random independent bits $r_1$ and $s_1$:
$$
    s = r_1 \oplus s_1.
$$
The secret for $\mathcal S$ will be $s$. Use any perfect ideal $k$ out of
$k$ threshold scheme to distribute $r_1$ among participants in the
intersection $\bigcap A_i$. That is, all participants from this set together
can reconstruct $r_1$, but no proper subset of them (or any any other subset
not containing all participants from $\bigcap A_i$) has any information on
$r_1$. Next, distribute $s_1$ according to the scheme $\mathcal S^*$. 
Participants in the intersection $\bigcap A_i^{\eps_i}$ (not all $\eps_i$
are equal) will be able to recover the share what $\mathcal S^*$ assigns to
them using again an independent realization of a perfect ideal $k$ out of
$k$ threshold scheme (different $k$, of course). It is clear that qualified
subsets of $\mathcal U_n$ can recover the secret $s$, while unqualified
subsets have no information on it. The complexity of $\mathcal S$ is
$\max\{1,\sigma(\mathcal S^*)\}$ (here $1$ comes from distributing $r_1$) as
threshold schemes have complexity $1$. This establishes $\sigma(\mathcal
U_n)\le\sigma(\mathcal U^*_n)$ as required.
\end{proof}

\begin{claim}
$\sigma(\mathcal U_2)=1$.
\end{claim}
\begin{proof}
In $\mathcal U^*_2$ there are $2^2-2 = 2$ participants, and both of them
form a one-element qualified set. Simply give the secret to both of them. 
\end{proof}

\section{Upper bound}

In this section we define a perfect scheme realizing $\mathcal U^*_n$ by
recursion on $n$. Using this scheme we establish an upper bound for the
complexity of the universal scheme. We call participants $p$ and $q$ of
$\mathcal U^*_n$ {\it equivalent} if there is a permutation $\pi$ on the
full set of participants which sends $p$ to $q$ and maps all qualified
subsets into qualified subsets, i.e., $\pi$ is an automorphism of the
structure $\mathcal U^*_n$. It is easy to see that $p$ and $q$ are
equivalent if and only if both of them are members of exactly the same
number of minimal qualified sets. Thus there are $n-1$ equivalence classes,
one for each $i=1,\dots,n-1$.

Suppose $\mathcal S_n$ realizes $\mathcal U^*_n$. Using the standard
symmetrization technique we may assume that $\mathcal S_n$ is {\it fully
symmetrical}, that is, equivalent participants receive shares of the same
size. We denote by $f_n(i)$ this common size of the shares of those who are
in exactly $i$ qualified subsets divided by the size of the secret. The
complexity of the scheme $\mathcal S_n$ is then
$$
    \sigma(\mathcal S_n) = \max \{ f_n(i) \,:\, 1\le i < n \}.
$$
Let us also define $f_n(0)=0$ and $f_n(n)=1$.

\begin{lemma}\label{lemma:induction}
Suppose $\mathcal S_n$ realizes $\mathcal U^*_n$ with $f_n(i)$ as defined
above. Then there is a perfect scheme $\mathcal S_{n+1}$ realizing
$\mathcal U^*_{n+1}$ such that
\begin{equation}\label{eq:recursion}
    n\cdot f_{n+1}(i) = (n+1-i)\cdot f_n(i) + i\cdot f_n(i-1) ~~\mbox{ for
all }~~ 0 < i \le n.
\end{equation}
\end{lemma}
\begin{proof}
We will use Stinson's decomposition technique from \cite{stinson}. Given
$U^*_{n+1}$ define $n+1$ access structures $\Gamma_1, \dots, \Gamma_{n+1}$
such that the set of participants is the same, while in $\Gamma_j$ the
$j$-th minimal qualified subset is left out. By the assumption and by Claim
\ref{claim:general} above there are perfect schemes realizing $\Gamma_j$
such that participants in $i$ out of the $n$ qualified subsets of $\Gamma_j$
receive $f_n(i)$-size shares for all $0\le i\le n$.

Now use all of these $n+1$ schemes simultaneously. Each (minimal) qualified
subset of $U^*_{n+1}$ recovers exactly $n$ out of the $n+1$ distributed
secrets, while an unqualified subset has no information on any of those
secrets. Using Stinson's trick from \cite{stinson} we can define a secret of
size $n$ and $n+1$ shadows of size $1$ each so that the secret can be
determined by any $n$ of the shadows. Thus the composite scheme will
distribute a secret of size $n$. A participant who is in exactly $i$ of the
minimal qualified subsets will receive shares of size $f_n(i-1)$ from
$\Gamma_j$ when $j$ is one of the minimal qualified subsets he is in ($i$ in
total), and shares of size $f_n(i)$ from $\Gamma_j$ when $j$ is one of the
minimal qualified subsets he is not in ($n+1-i$ in total). This establishes
the claim of the lemma.
\end{proof}

\begin{lemma}\label{lemma:value}
Suppose $f_n(0)=0$, $f_n(n)=1$, $f_2(1)=1$ and for $n\ge 2$ the function 
$f_{n+1}$ satisfies the recursion in {\rm(\ref{eq:recursion})}. Then
\begin{equation}\label{eq:closed}
    f_n(i) = (n-i)\big(h(n)-h(n-i)\big) ~~\mbox{ for all }~~ 0\le i < n,
\end{equation}
where $h(1)=0$ and  $h(n)=\sum_{0<j<n}1/j$ otherwise.
\end{lemma}
\begin{proof}
When $i=0$ then, by definition, $f_n(0)=0$, and (\ref{eq:closed}) also yields
$0$. (\ref{eq:closed}) also gives
$$
  f_2(1)= h(2)-h(1)=1-0=1.
$$
Suppose (\ref{eq:closed}) holds for $n\ge 2$ and $0<i<n$. Equation (\ref{eq:recursion}) can be rewritten as
\begin{equation}\label{eq:modrecursion}
\frac{f_{n+1}(i)}{n+1-i} = \frac{n-i}{n}\cdot\frac{f_n(i)}{n-i} + 
     \frac{i}{n}\cdot\frac{f_n(i-1)}{n+1-i}
\end{equation}
Now $h(n+1)=h(n)+1/n$, and $h(n+1-i)=h(n-i)+1/(n-i)$. Consequently
$$
\begin{array}{r@{\;}c@{\;}l}
\displaystyle    h(n+1) &=& \displaystyle h(n)+\frac1n = \frac i n \cdot h(n)+\frac{n-i}n
\cdot h(n)
+\frac1n,\\[8pt]
\displaystyle    h(n+1-i) &=& \displaystyle  \frac i n\cdot h(n+1-i) +
\frac{n-i}n\cdot\big(h(n-i)+\frac1{n-i}\big),
\end{array}
$$
that is
$$
   h(n+1)-h(n+1-i) = \frac{n-i}n\big(h(n)-h(n-i)\big) + 
      \frac i n \big(h(n)-h(n+1-i)\big).
$$
This shows that the function $h(n)-h(n-i)$ satisfies the recursion
(\ref{eq:modrecursion}), that is $f_n(i)=(n-i)\big(h(n)-h(n-i)\big)$ for all
$n$ and $i$ by induction.
\end{proof}

\begin{theorem}
The complexity of $\mathcal U_n$ is asymptotically at most $n/e$ where $e=2.7182\dots$ is the
Euler number.
\end{theorem}

\begin{proof}
By Lemma \ref{lemma:induction} there is a scheme $\mathcal S_n$ realizing
$\mathcal U_n$ with complexity
$$
    \sigma(\mathcal S_n) = \max \{ f_n(i)\,:\,\ 1\le i < n \}
$$
where, by Lemma \ref{lemma:value} $f_n(i)=(n-i)\big(h(n)-h(n-i)\big)$. Now
\begin{equation}\label{eq:derivative}
    f_n(i+1)-f_n(i)= h(n-i)+1- h(n).
\end{equation}
As $h(n-i)$ is strictly decreasing, $f_n$ increases while $h(n-i)\ge
h(n)-1$, and decreases after this value of $i$, and takes its maximum when 
(\ref{eq:derivative}) changes sign. Using the approximation $h(x) = \gamma + 
\log(x-1/2)+O(x^{-2})$ from \cite{mathworld}
where $\gamma$ is the Euler-Mascheroni constant, this happens when
$$
    n-i = \frac n e + \frac{e-1}{2e} + O(1/n),
$$
and then the maximal value of $f_n$ taken at this $i$ is
$$
    f_n(i)=
    (n-i)\big(1+O(1/n)\big) = \frac ne  + \frac{e-1}{2e} + O(1/n).
$$
From here the theorem follows.
\end{proof}
A more careful analysis shows that
$$
\max_i f_n(i) < \frac ne + \frac{e-1}{2e} + \frac 1{3n}< \frac ne + 0.5
$$
for all $n\ge 2$ which gives a non-asymptotic estimate on the complexity of
$\mathcal S_n$.

\section{Lower bound}
As mentioned in the Introduction, for the lower bound we use a variation of
the {\it independent sequence method} from \cite{tight}. For a general
description how the method works, see, e.g., \cite{tight, csirmaz}. Briefly,
for any subset $A$ of the participant one considers the {\it relative
entropy $f(A)$}, which is simply the entropy of all the shares given to
members of $A$ divided by the entropy of the secret. The function $f$
satisfies a certain set of linear inequalities derived from the Shannon
inequalities for the entropy function. For example, $f(\emptyset)=0$, 
for arbitrary subsets $A$ and $B$ of the participants
$$
    f(AB) \le f(A)+f(B),
$$
or, $f(B)\le f(A)$ when $B\subseteq A$. This latter inequality is called {\it monotonicity}.
The {\it strict monotonicity property} says that whenever
$B$ is independent and $B\subseteq A$ is qualified then not
only $f(A)$ is larger than $f(B)$, but the difference is at least one:
$f(B)+1\le f(A)$. As a particular case,
\begin{equation}\label{eq:additive}
    f(A) \le f(a_1)+f(a_2)+\cdots+f(a_k)
\end{equation}
when $A=\{a_1,\dots,\allowbreak a_k\}$.
Here and the in the sequel, as usual, we write $f(a)$ instead of $f(\{a\})$,
and write $AB$ and $Aa$ instead of $A\cup B$ and $A\cup\{a\}$, respectively.

The {\it entropy method} works as follows. To prove that the complexity of
an access structure is at least $\kappa$ it is enough to show that for all
non-negative functions $f$ satisfying all of the above
inequalities there is a participant $p$ with $f(p)\ge \kappa$. 

The {\it independent sequence method} relies on the following lemma. For the
sake of self-containment we supply a proof.

\begin{lemma}\label{lemma:basicentropy}
Suppose $A$ is qualified, $B\subseteq A$, $p$ is a participant not in $A$, 
and for some $B
\subseteq C \subseteq A$, $C$ is independent, while $Cp=C\cup\{p\}$ is
qualified. Then $f(A)-f(B)\ge\allowbreak f(Ap)-f(Bp) + 1$.
\end{lemma}
\begin{proof}
As the function $f$ is the relative entropy, it is {\em
submodular}, i.e. $f(X)+f(Y) \ge f(X\cup Y)+f(X\cap Y)$ for any two subsets
$X$ and $Y$. In fact, this is the same inequality which says that the mutual
conditional information is non-negative, see \cite{csiszar-korner}. If $s$
denotes the secret, then $f(Xs)=f(X)$ if $X$ is qualified, as $X$ can
determine the secret, furthermore $f(Xs)=f(X)+1$ if $X$ is unqualified,
expressing the fact that the information members of $X$ have is {\em
independent} of the secret.

By the assumptions of the lemma, $A$ and $Cp$ are qualified, while $C$ is
not, thus by the submodularity
\begin{equation}\label{eqlemma:1}
   f(A)+f(Cp) = f(As)+f(Cps) \ge f(Aps)+ f(Cs) = f(Ap)+(f(C)+1),
\end{equation}
as $C\subseteq A$.
Using submodularity for the sets $C$ and $Bp$ yields
\begin{equation}\label{eqlemma:2}
  f(C)+b(Bp) \ge f(Cp)+f(B),
\end{equation}
here we used that $B\subseteq C$. The sum of (\ref{eqlemma:1}) and
(\ref{eqlemma:2}) gives claim of the lemma immediately.
\end{proof}

The next lemma summarizes the method itself.

\begin{lemma}[Independent sequence method]\label{lemma:independent-sequence}
Suppose $A_0$ is a qualified set, $B_0$ is the empty set,
and $b_1$, $\dots$, $b_n$ are participants not in $A_0$. Let
$B_i=\{b_1,\allowbreak\dots,b_i\}$, and $A_i=A_0\cup B_i$. Suppose that for all 
$0\le i< n$ there is a subset $C_i\subset A_0$ such that $B_iC_i$ is
independent, while $B_{i+1}C_i = b_{i+1} B_i C_i$ is 
qualified. Then $f(A_0)\ge n$.
\end{lemma}
In the usual terminology the sequence $A_0$, $A_1$, $\dots$ is {\it made
independent} by the sequence $B_0$, $B_1$, $\dots$ from where the method got its
name.
\begin{proof}
By the monotonicity of the function $f$,  $f(A_n)-f(B_n)\ge 0$ as 
$A_n\supseteq B_n$.
By Lemma \ref{lemma:basicentropy},
$$
     f(A_i)-f(B_i) \ge f(A_{i+1})-f(B_{i+1}) + 1
$$
which is shown by the set $C=B_iC_i$ and the participant $p=b_{i+1}$. Putting all of these inequalities together we get
$$
    f(A_0)-f(B_0) \ge n.
$$
As $B_0$ is the emptyset, $f(B_0)=0$ and the claim of the lemma follows.
\end{proof}

\medskip

As it was remarked in the Introduction, among all access structures with $n$
minimal qualified subsets, $\mathcal U_n$ has the highest complexity. 
Consequently to show that $\mathcal U_n$ has complexity at least $\kappa$ it
is enough to find {\it any} access structure with $n$ minimal qualified
subsets with complexity ${}\ge\kappa$. This is exactly what we will do.

\begin{theorem}\label{thm:lower-bound}
For each $n\ge 2$ there is an access structure with $n$ minimal qualified 
subsets and complexity at least $n/(1+\log_2 n)$.
\end{theorem}

\begin{proof}
Let $k\ge 1$ to be chosen later and $X$ be a set with exactly $k$ elements.
Enumerate the proper subsets of $X$ in {\it decreasing order} as $C_0$,
$C_1$, $\dots$, $C_{\ell-1}=\emptyset$ where $\ell=2^k-1$. ``Decreasing'' means
that sets with larger index does not have more elements; in particular 
$C_i\not\subset C_j$ whenever $0\le i < j <\ell$.

Let $b_1$, $\dots$, $b_\ell$ be new participants not in $X$, and define
$B_0=\emptyset$, and for $1\le i \le \ell$ let $B_i$ be the set
$\{b_1,\allowbreak b_2,\dots,b_i\}$. There will be $n\le\ell$ minimal qualified
subsets: $B_1C_0$, $B_2C_1$, $\dots$, $B_nC_{n-1}$. These sets can indeed
form a collection of minimal qualified subsets as none of them is a subset
of any other. Also we remark that the sets $B_iC_i$ for $0\le i < n$ are
{\it independent} as none of $B_{j+1}C_j$ is a subset of $B_iC_i$. Indeed,
if $j\ge i$ then $B_{j+1}$ is not a subset of $B_i$, and if $j<i$ then $C_j$
is not a subset of $C_i$. Thus we can apply Lemma
\ref{lemma:independent-sequence} with $A_0=X$ giving
$$
    f(X)\ge n.
$$
As $X$ has $k$ elements, say $a_1$, $\dots$, $a_k$, applying property
(\ref{eq:additive}) we get
$$
    f(a_1)+\cdots+f(a_k) \ge f(X) \ge n,
$$
meaning that some participant in this access structure has an $f$-value at
least $n/k$. Thus the complexity of the structure is also at least $n/k$. 
Our aim is to choose $k$ as small as possible. We have only a single
restriction, namely that $n$ should not exceed the value $\ell=2^k-1$. Thus
there is always an appropriate $k$ which is less than or equal to $(1+\log_2n)$,
proving the theorem.
\end{proof}

As a consequence of Theorem \ref{thm:lower-bound} the universal structure
$\mathcal U_n$ also has complexity at least $n/(1+\log_2n)$, proving that the
complexity is at least $n/\log_2n$ asymptotically.

\section*{Acknowledgment}
The author would like to thank one of the referees for pointing out an
omission in Lemma \ref{lemma:basicentropy}, and for the suggestions which
helped to improve the paper.

\end{document}